\documentclass[review]{elsarticle}

\usepackage{lineno,hyperref}
\modulolinenumbers[5]

\journal{Journal of \LaTeX\ Templates}

\usepackage{amsmath,amssymb,amsthm}
\newtheorem{theo}{Theorem}[section]

\newtheorem{lema}{Lemma}[section]

\newcommand{\dx}{\mathrm d^4x}
\newcommand{\FF}{\mathcal{F}}

\newcommand{\Natural}{\mathbb{N}}
\newcommand{\ve}{\varepsilon}

\begin{document}

\begin{frontmatter}

\title{Cosmological solutions in modified gravity with monomial nonlocality}

\author[mymainaddress]{Ivan Dimitrijevic}
\ead{ivand@math.rs}

\address[mymainaddress]{Studentski trg 16, 11000 Belgrade Serbia}

\begin{abstract}
We consider cosmological properties of modified gravity with nonlocal term $ R^p\FF(\Box)R^q$ in its Lagrangian.
Equations of motion are presented. For the flat FLRW metric, and some particular values of natural numbers $p$ and $q$ cosmological solutions of the form $a(t)= C e^{- \frac \gamma{12}t^2}$ are found.
\end{abstract}


\end{frontmatter}

\linenumbers

\section{Introduction}

Modern theory of gravity is general theory of relativity (GR), which was founded by Einstein one hundred years ago and has been successfully confirmed for the Solar System. It is given by the Einstein equations
of motion for gravitational field  $g_{\mu\nu}$:
$
R_{\mu\nu} - \frac{1}{2} R g_{\mu\nu} = 8 \pi G T_{\mu\nu},
$
which can be derived from the Einstein-Hilbert action $S =
\frac{1}{16 \pi G} \int R \sqrt{-g} \dx + \int \mathcal{L}_{mat} \sqrt{-g} \dx$,  where $g = det (g_{\mu\nu})$ and units are chosen in such way that $c = 1$.

Despite all its successes, GR is not a final theory of gravity. There are  many its modifications,
which are motivated by quantum gravity, string theory, astrophysics and cosmology (for a  review, see \cite{clifton}).
One of very promising directions of research is {\it nonlocal modified gravity} and its applications to cosmology (as a review, see
 \cite{odintsov} and \cite{Woodard:2014iga}). To solve cosmological Big Bang singularity, nonlocal gravity with replacement $R \to
 R + C R \mathcal{F}(\Box)R$  in the Einstein-Hilbert action was proposed in \cite{biswas0}. This nonlocal model is further elaborated
 is the series of papers \cite{biswas,biswas:2011ar,koshelev2,DDGR1,DDGR2,biswas:2013cha}.

In this  paper we consider the  action
\begin{equation} \label{action}
  S_{pq} = \int \left( \frac {R-2\Lambda}{16\pi G} + R^p\FF\left(\frac\Box{M^2}\right)R^q \right) \sqrt{-g} \dx
\end{equation}
 where $R$ is scalar curvature, $ \mathcal{F}(\Box)= \displaystyle \sum_{n =0}^{\infty} f_{n}\Box^{n}$  is an analytic function  of the d'Alembert-Beltrami operator $\Box = \frac{1}{\sqrt{-g}} \partial_{\mu}
\sqrt{-g} g^{\mu\nu} \partial_{\nu},$ $\, g = det(g_{\mu\nu})$, $M$ is a characteristic scale and $p$ and $q$ are natural numbers. For simplicity we take $M=1$. At the end of this paper we briefly discuss the limit when $M\to +\infty$.   
In the paper \cite{DDGR6} action \eqref{action} was introduced and  constant scalar curvature cosmological solutions were obtained. Also, perturbations around de Sitter background were discussed in \cite{DDGKR}.


\section{Equations of motion}
Variation of the action \eqref{action} with respect to metric yields the equations of motion in the form
\begin{equation} \label{eom}
  -\frac 12 g_{\mu\nu} R^p\FF(\Box)R^q + R_{\mu\nu} W_{pq} - K_{\mu\nu} W_{pq} + \frac 12 \Omega_{pq\,\mu\nu} = - \frac {G_{\mu\nu}+ \Lambda g_{\mu\nu}}{16\pi G},
\end{equation}
where
\begin{equation}\begin{aligned}
  W_{pq} &= p R^{p-1}\FF(\Box)R^q + q R^{q-1} \FF(\Box)R^p, \\
   K_{\mu\nu} &= \nabla_\mu \nabla_\nu - g_{\mu\nu}\Box, \\
  \Omega_{pq\,\mu\nu} &= \sum_{n=1}^\infty f_n \sum_{l=0}^{n-1} \Big(g_{\mu\nu} \nabla^\lambda \Box^l R^p \nabla_\lambda \Box^{n-1-l} R^q + g_{\mu\nu}  \Box^l R^p \Box^{n-l}R^q \\
  &-2  \nabla_\mu\Box^l R^p \nabla_\nu \Box^{n-1-l}R^q \Big), \\
\end{aligned}\end{equation}

Detailed derivation of the above equations can be found in \cite{DDGR6}.

In this paper Friedmann-Lema\^{\i}tre-Robertson-Walker (FLRW) metric
$ds^2 = - dt^2 + a^2(t)\big(dr^2 + r^2 d\theta^2 + r^2 \sin^2 \theta d\phi^2\big)$ is used. The signature of a metric is $(1,3)$ and the sign of a curvature tensor is chosen such that 
\begin{align}
  R_{\mu\nu\alpha}^\beta &= \partial_\nu \Gamma_{\mu\alpha}^\beta - \partial_\mu \Gamma_{\nu\alpha}^\beta + \Gamma_{\mu\alpha}^\lambda \Gamma_{\nu\lambda}^\beta - \Gamma_{\nu\alpha}^\lambda \Gamma_{\mu\lambda}^\beta,\\
  R_{\mu\nu} &= R_{\mu\lambda\nu}^\lambda.
\end{align}
Scalar curvature is $R = R_{\mu\nu} g^{\mu\nu} = 6 \left (\frac{\ddot{a}}{a} + \frac{\dot{a}^{2}}{a^{2}} \right )$ and $\Box
h(t)= - \partial_t^2 h(t) - 3 H \partial_t h(t) ,$ where $H = \frac{\dot{a}}{a}$ is the Hubble parameter.

\begin{lema}
  If the metric is chosen to be FLRW, then the system \eqref{eom} has two linearly independent equations.
\end{lema}
\begin{proof}
Note, that for the functions that only depend on time we have $K_{\mu\nu}=0$ when $\mu\neq \nu$. Also $\Omega_{pq \,\mu\nu} =0$ for $\mu\neq \nu$.It means that system \eqref{eom} has four nontrivial equations. Equations with indices  $\mu\nu$ equal to $11$, $22$ and $33$ can be rewritten as
\begin{equation} \begin{aligned}
  &16\pi G g_{ii} \Big(-\frac 12  R^p \FF(\Box)R^q + \left(\frac{\ddot a}{a}+ 2\left(\frac{\dot a}{a}\right)^2\right) W_{pq} - \left( \ddot W_{pq} + 2\frac{\dot a}a \dot W_{pq} \right) \\
  &+ \frac 12 \sum_{n=1}^\infty f_n \sum_{l=0}^{n-1} \left(\nabla^\lambda \Box^l R^p \nabla_\lambda \Box^{n-1-l} R^q + \Box^l R^p \Box^{n-l}R^q\right) \Big) = g_{ii}\left (\frac{2\ddot{a}a+\dot a^2}{a} - \Lambda\right ).
\end{aligned} \end{equation}

These equations are clearly proportional to each other and thus we have altogether two  independent equations. The most convenient choice is to use trace and $00$ equations, which are respectively
\begin{align}\label{trace}
  -2 R^p\FF(\Box)R^q + R W_{pq} + 3\Box W_{pq} + \frac 12 \Omega_{pq} &=  \frac {R - 4\Lambda}{16\pi G}, \\
  \frac 12  R^p\FF(\Box)R^q + R_{00} W_{pq} - K_{00} W_{pq} + \frac 12 \Omega_{pq\,00} &= - \frac {G_{00} - \Lambda}{16\pi G}, \label{eom:00} \\
  \Omega_{pq} &= g^{\mu\nu}\Omega_{pq\, \mu\nu}.
\end{align}

\end{proof}

At first, we investigate how does equations \eqref{eom} change when parameters $p$ and $q$ replace their places in the action \eqref{action}. To this end, the following lemma holds.
\begin{lema}\label{lem:symetric}
If we consider actions $S_{pq}$, given in \eqref{action}, and $S_{qp}$. The corresponding equations of motion are equivalent.
\end{lema}
\begin{proof}
  Equations of motion, given by equation \eqref{eom}, for actions $S_{pq}$ and $S_{qp}$ read
\begin{align}
  -\frac 12 g_{\mu\nu} R^p\FF(\Box)R^q + R_{\mu\nu} W_{pq} - K_{\mu\nu} W_{pq} + \frac 12 \Omega_{pq\,\mu\nu} &= - \frac {G_{\mu\nu}+ \Lambda g_{\mu\nu}}{16\pi G}, \\
  -\frac 12 g_{\mu\nu} R^q\FF(\Box)R^p + R_{\mu\nu} W_{qp} - K_{\mu\nu} W_{qp} + \frac 12 \Omega_{qp\,\mu\nu} &= - \frac {G_{\mu\nu}+ \Lambda g_{\mu\nu}}{16\pi G}.
\end{align}

Since $W_{pq}$ and $W_{qp}$ coincide, subtraction of the last two equation yields
\begin{align}
    (R^p\FF(\Box)R^q- R^q\FF(\Box)R^p) =  \sum_{n=1}^\infty f_n \sum_{l=0}^{n-1} \left(\Box^l R^p \Box^{n-l}R^q - \Box^l R^q \Box^{n-l} R^p\right).
\end{align}
The terms obtained for $l=0$ on the right hand side of the last equation exactly match the left hand side, thus we are left with
\begin{align}
    \sum_{n=1}^\infty f_n \sum_{l=1}^{n-1} \left(\Box^l R^p \Box^{n-l}R^q - \Box^l R^q \Box^{n-l} R^p\right) =0.
\end{align}
This equation can be proved by changing the summation index $l\to n-l$ in one of the terms, which completes the proof. It is worth noting that we assume that total derivatives terms arising from the partial integration vanish.
\end{proof}
\section{The scale factor}
In this paper we consider the scale factor in the form
\begin{equation} \label{scale_factor}
  a(t) = C e^{- \frac \gamma{12}t^2}.
\end{equation}
This form of scale factor has been introduced, as a solution of the $p=q=1$ case in the paper \cite{koshelev2}. As a solution of $R^2$ gravity the scale factor \eqref{scale_factor} has been introduced in \cite{ruzmaikin} and further studied in papers \cite{starobinsky1,starobinsky2}. The correspondence between the $R^2$ gravity model and nonlocal model for $p=q=1$ has been discussed in \cite{Koshelev:2014voa}.  The present paper generalizes this result to various values of parameters $p$ and $q$ in  action $S_{pq}$. It is worth noting that $\gamma =0$ gives the Minkowski spacetime, and it is a solution of equations of motion \eqref{eom} for $\Lambda =0$. The following analysis does not depend on the sign of $\gamma$ and gives expanding ($\gamma <0$) and contracting ($\gamma>0$) models.

The Hubble parameter and scalar curvature are linear and quadratic functions in cosmic time $t$, respectively
\begin{equation} \label{HR}
  H(t) = -\frac 16 \gamma t, \qquad R(t) = \frac 13\gamma(\gamma t^2 -3).
\end{equation}
By direct calculation one can show that for any natural number $p$, $\Box R^p$ is a linear combination of $R^p$, $R^{p-1}$ and $R^{p-2}$, i.e.
\begin{equation}
  \Box R^p = p \gamma R^p - \frac p3(4p-5)\gamma^2 R^{p-1} - \frac 43 p(p-1)\gamma^3 R^{p-2}. \label{ansatz}
\end{equation}

\begin{lema} \label{lemma:matrix}
  For fixed value of parameter $\gamma$, an therefore fixed values of Hubble parameter $H$ and scalar curvature $R$,
  consider the space $P_p(R)$ of all polynomials of degree at most $p$ in $R$ and its base $v_p = \left(
           \begin{array}{ccccc}
             R^p & R^{p-1} & \ldots & R & 1 \\
           \end{array}
         \right)^T$. Operator $\Box$ is a linear operator on $P_p(R)$. The matrix of the operator $\Box$ in the basis $v_p$ is
\begin{align}
  M_p &= \small{\gamma}\left(
           \begin{array}{cccccc}
             p  & \frac p3(5-4p)\gamma & \frac 43p(1-p)\gamma^2 & 0 & \ldots & 0 \\
             0 & p-1 & \frac {p-1}3(9-4p)\gamma & \frac 43(1-p)(p-2)\gamma^2 & \ldots & 0 \\
             \vdots & \vdots & \vdots & \ddots & \vdots \\
             0 & 0 & 0 & \ldots & 1 & \frac{\gamma}3 \\
             0 & 0 & 0 & \ldots & 0 & 0 \\
           \end{array}
         \right).
\end{align}
\end{lema}
\begin{proof}
  Since $\Box h(t)= - \partial_t^2 h(t) - 3 H \partial_t h(t)$ it is clear that $\Box$ is a linear operator. It remains to prove that $\Box R^s \in P_p(R)$ for all $0\leq s \leq p$. For $s=0$, we have $\Box 1 = 0 \in P_p(R)$. For $s=1$ equation \eqref{ansatz} becomes $\Box R = \gamma R + \frac{\gamma^2}3$ which is a linear polynomial in $R$ and therefore an element of $P_p(R)$. For $2\leq s \leq p$ equation \eqref{ansatz} gives us that $\Box R^s$ is a polynomial of degree $s$ in $R$ and hence element of $P_p(R)$.
  Again, from equation \eqref{ansatz} we obtain that the matrix $M_p$ has the form given in lemma.
\end{proof}

As a consequence of the lemma \ref{lemma:matrix} $\Box^n R^p$ is expressible as a polynomial in $R$ of degree $p$. Let $F_p$ be the matrix of the operator $\FF(\Box)$,
\begin{align}
  F_p &= \sum_{n=0}^\infty f_n M_p^n = \FF(M_p).
\end{align}

\section{The general case}

Lemma \ref{lem:symetric} allow us to assume that $p\geq q$ in the following sections. Also, it is worth noting that $W_{pq}$ is a polynomial of degree $p+q-1$ in $R$, as well as
\begin{align}
  \FF(\Box) R^p &= e_p F_p v_p, \\
  W_{pq} &= p R^{p-1} e_q F_q v_q + q R^{q-1} e_p F_p v_p, \\
  \Box W_{pq} &= -\frac 43 \gamma^2 (R+\gamma)W_{pq}'' -\frac 23 \gamma^2 W_{pq}',\label{box}\\
  K_{00} W_{pq} &= \gamma(R+\gamma) W_{pq}',
\end{align}
and $'$ denotes derivation {\it{wrt}} $R$. The $\Omega_{pq}$ and $\Omega_{pq\, 00}$ terms are also polynomials in $R$ with degrees $p+q-1$ and $p+q$ respectively
\begin{align}
  \Omega_{pq} &= -2S_1 + 4S_2,\\
  \Omega_{pq\ 00} &= -S_1 - S_2,\\
  S_1 &= \frac 43 \gamma^2 (R+\gamma) \sum_{n=1}^\infty f_n \sum_{l=0}^{n-1} e_p M_p^l D_p v_p e_q M_q^{n-1-l} D_q v_q, \\
  S_2 &= \sum_{n=1}^\infty f_n \sum_{l=0}^{n-1} e_p M_p^l v_p e_q M_q^{n-l} v_q.
\end{align}

Hence, equations \eqref{trace} and \eqref{eom:00} are polynomial type in $R$ and their degree is $p+q$. Let us look at the highest order coefficient. If $p\neq q$ we obtain the following two equations
\begin{align}
p \FF(q\gamma)(q-p+2)+q \FF(p\gamma)(q-p-2) &=0,\\
(-q-\frac{1}{2}p(q-p))\FF(q\gamma)+(-\frac{1}{2}q(q-p)+q) \FF(p\gamma) &=0.
\end{align}
These equations are linearly dependent for any values of parameters $p \neq q$.

In the other case $p=q$ the previous system becomes
\begin{align}
(p-1) \FF(p\gamma)+ p\gamma \FF'(p\gamma) &=0,\\
-\frac 12(p-1) \FF(p\gamma)-\frac 12 p\gamma \FF'(p\gamma) &=0.
\end{align}
We conclude that, this case also yields linearly dependent equations.

\section{Particular cases}

\subsection{Case $p=1, q=1$}
At the beginning we discuss the simplest case $p=q=1$. The solution of the form $a(t) = C \exp( \Lambda t^2)$ was already obtained by Koshelev and Vernov in \cite{koshelev2}.

\begin{theo}
  If the scale factor has the form $a(t)=C e^{-\frac \gamma{12}t^2}$ and $p=q=1$, then the system \eqref{trace}, $\eqref{eom:00}$ is satisfied iff $\gamma = -12\Lambda$, $\FF'(\gamma) = 0$ and $f_0 = \frac {3\kappa}{2\gamma} - 8 \FF(\gamma)$ where $\kappa = \frac 1{16\pi G}$.
\end{theo}
\begin{proof}
Trace and $00$ equations are written as
\begin{align}
  T_2 R^2 + T_1 R + T_0 &= 0, \label{trace:11}\\
  Z_2 R^2 + Z_1 R + Z_0 &= 0, \label{eom:00:11}
\end{align}

where

\begin{equation}
\begin{aligned}
  T_0 &= \frac{2}{9} \left(-5 \FF'(\gamma ) \gamma ^3+8 \FF(\gamma ) \gamma ^2+f_0 \gamma ^2+18 \kappa \Lambda \right), \\
  T_1 &= \frac{1}{3} \left(-3 \kappa +16 \gamma  \FF(\gamma )+2 \gamma  f_0\right), \\
  T_2 &= 2 \gamma  \FF'(\gamma ), \\
  Z_0 &= \frac{1}{36} \left(-26 \FF'(\gamma ) \gamma ^3-64 \FF(\gamma ) \gamma ^2-8 f_0 \gamma ^2+9 \kappa  \gamma -36 \kappa  \Lambda \right), \\
  Z_1 &= \frac{1}{12} \left(-12\FF'(\gamma ) \gamma ^2-16 \FF(\gamma ) \gamma -2 f_0 \gamma +3 \kappa \right), \\
  Z_2 &= -\frac{1}{2} \gamma  \FF'(\gamma ).
\end{aligned}
\end{equation}

At first, note that $T_0 + 4Z_0 = 4\gamma Z_1$, $T_1+4Z_1 = 8\gamma Z_2$, $T_2 + 4Z_2=0$ and hence equations \eqref{trace:11} and \eqref{eom:00:11} are equivalent. Therefore it is sufficient only to look at the trace. On the other hand from  equation \eqref{HR} we see that $R$ is a quadratic function in time and hence equation \eqref{trace:11} is satisfied for all values of time $t$ iff $T_0 = T_1 = T_2 =0$.  yields the linear system in $f_0$, $\FF(\gamma)$ and $\FF'(\gamma)$. It is consistent only for $\gamma = -12 \Lambda$ and the solution is
\begin{equation}
  \FF'(\gamma) = 0, \qquad f_0 = \frac {3\kappa}{2\gamma} - 8 \FF(\gamma).
\end{equation}

\end{proof}

\subsection{Case $(p,q)\neq (1,1)$}

As a consequence of lemma \ref{lemma:matrix} we can write
\begin{equation}
  R^p = e_p v_p, \quad \Box^n R^p = e_p M_p^n v_p, \quad \FF(\Box) R^p = e_p F_p v_p,
\end{equation}
where $e_p$ are the coordinates of $R^p$ in basis $v_p$ and $D_p$ be a matrix such that $\dfrac{\partial v_p}{\partial R} = D_p v_p$ and
$p\in \Natural$. Therefore the system \eqref{trace}, $\eqref{eom:00}$ can be written as
\begin{align}\label{trace:2}
&\frac {R - 4\Lambda}{16\pi G} = R W_{pq} -4 \gamma^2 (R+\gamma)W_{pq}'' -2 \gamma^2 W_{pq}'  -2 e_p v_p e_q F_q v_q  -S_1 + 2S_2, \\
&\frac{\Lambda -G_{00}}{16\pi G} = \frac 12  e_p v_p e_q F_q v_q + \frac \gamma4(\gamma-R) W_{pq} - \gamma(R+\gamma) W_{pq}' - \frac 12 (S_1+S_2),\label{eom:00:2}
\end{align}
where
\begin{align}
  W_{pq} &= e_p D_p v_p e_q F_q v_q + e_q D_q v_q e_p F_p v_p, \\
  S_1 &= \frac 43 \gamma^2 (R+\gamma) \sum_{n=1}^\infty f_n \sum_{l=0}^{n-1} e_p M_p^l D_p v_p e_q M_q^{n-1-l} D_q v_q, \\
  S_2 &= \sum_{n=1}^\infty f_n \sum_{l=0}^{n-1} e_p M_p^l v_p e_q M_q^{n-l} v_q.
\end{align}

\begin{theo} \label{theo:tz}
  Let
  \begin{align}
 T&= -2 e_p v_p e_q F_q v_q + R W_{pq} -4 \gamma^2 (R+\gamma)W_{pq}'' -2 \gamma^2 W_{pq}'\nonumber \\
 &-S_1 + 2S_2 -  \frac {R - 4\Lambda}{16\pi G}, \\
 Z&= \frac 12  e_p v_p e_q F_q v_q + \frac \gamma4(\gamma-R) W_{pq} - \gamma(R+\gamma) W_{pq}' \nonumber\\
  &- \frac 12 (S_1+S_2) + \frac {G_{00} - \Lambda}{16\pi G},
\end{align}
then $T+4Z = 4\gamma Z'$.
  The equations \eqref{trace:2} and  $\eqref{eom:00:2}$ are equivalent.
\end{theo}
\begin{proof}
  To prove the first part of the theorem, direct calculations give
\begin{align}
  T+4Z &= \gamma W_{pq}- \gamma (R+3\gamma) W_{pq}' -4\gamma^2 (R+\gamma)W_{pq}'' -3S_1, \\
  4\gamma Z' &= 2\gamma \left(R^p \FF(\Box)R^q\right)' - \gamma W_{pq}- \gamma (R+3\gamma) W_{pq}' -4\gamma^2 (R+\gamma)W_{pq}'' \nonumber \\
  &-2\gamma\left(S_1 + S_2\right)'.
\end{align}
After some simplifications we are left with the following equation
\begin{align}\label{15.10.15:eq1}
  W_{pq} - \left(R^p \FF(\Box)R^q\right)' &= \frac 32 \gamma^{-1}S_1 - S_1' -S_2'.
\end{align}

Instead of expressing  $\Box^n R^p$ in basis $v_p$, and $\Box^n R^q$ in basis $v_q$ it is more convenient to express all the terms in basis $v_{p+q}$. Therefore let $\ve_p$ and $\ve_q$ be a coordinates of $R^p$ and $R^q$ respectively in basis $v_{p+q}$.
 Then the above equation becomes
\begin{align}
\sum_{n=1}^\infty f_n \varepsilon_p\left( M_{p+q}^n v_{p+q} \varepsilon_q- v_{p+q} \varepsilon_q M_{p+q}^n \right)D_{p+q} v_{p+q} = \sum_{n=1}^\infty f_n q_n
\end{align}

where
\begin{equation}
\begin{aligned}
  q_n = \sum_{l=0}^{n-1}\Big(& 2\gamma(R+\frac\gamma3) \varepsilon_p M_{p+q}^l D_{p+q} v_{p+q}\varepsilon_q M_{p+q}^{n-l-1} D_{p+q} v_{p+q} \\
  & -  \frac 43\gamma^2(R+\gamma) \varepsilon_p M_{p+q}^l D_{p+q}^2 v_{p+q}\varepsilon_q M_{p+q}^{n-l-1} D_{p+q} v_{p+q} \\
  & -  \frac 43\gamma^2(R+\gamma) \varepsilon_p M_{p+q}^l D_{p+q} v_{p+q}\varepsilon_q M_{p+q}^{n-l-1} D_{p+q}^2 v_{p+q} \\
  & -   \varepsilon_p M_{p+q}^l D_{p+q} v_{p+q}\varepsilon_q M_{p+q}^{n-l} v_{p+q} \\
  & -   \varepsilon_p M_{p+q}^l v_{p+q}\varepsilon_q M_{p+q}^{n-l} D_{p+q} v_{p+q} \Big).
\end{aligned}\end{equation}

It is sufficient to prove
\begin{equation}
  M_{p+q}^n v_{p+q} \varepsilon_q- v_{p+q} \varepsilon_q M_{p+q}^n  = q_n.
\end{equation}

Now, move the last term in $q_n$ to the left side and after some index relabeling one obtains
\begin{equation}
\begin{aligned}
  &\sum_{l=0}^{n-1} \varepsilon_p M_{p+q}^{l+1} v_{p+q}\varepsilon_q M_{p+q}^{n-l-1} D_{p+q} v_{p+q} \\
  &= \sum_{l=0}^{n-1}\Big(2\gamma(R+\frac\gamma3) \varepsilon_p M_{p+q}^l D_{p+q} v_{p+q}\varepsilon_q M_{p+q}^{n-l-1} D_{p+q} v_{p+q} \\
  & -  \frac 43\gamma^2(R+\gamma) \varepsilon_p M_{p+q}^l D_{p+q}^2 v_{p+q}\varepsilon_q M_{p+q}^{n-l-1} D_{p+q} v_{p+q} \\
  & -  \frac 43\gamma^2(R+\gamma) \varepsilon_p M_{p+q}^l D_{p+q} v_{p+q}\varepsilon_q M_{p+q}^{n-l-1} D_{p+q}^2 v_{p+q} \\
  & -   \varepsilon_p M_{p+q}^l D_{p+q} v_{p+q}\varepsilon_q M_{p+q}^{n-l} v_{p+q} \Big).
\end{aligned}\end{equation}

Let us introduce matrix function $\alpha$ by :
\begin{equation}
  \alpha(X) = \sum_{l=0}^{n-1} M_{p+q}^l X M_{p+q}^{n-l-1}.
\end{equation}
Then the previous equation becomes
\begin{equation}
\begin{aligned}
   \varepsilon_p M_{p+q} \alpha( v_{p+q}\varepsilon_q) D_{p+q} v_{p+q}&= 2\gamma(R+\frac\gamma3) \varepsilon_p \alpha( D_{p+q} v_{p+q}\varepsilon_q) D_{p+q} v_{p+q} \\
  & -  \frac 43\gamma^2(R+\gamma) \varepsilon_p \alpha( D_{p+q}^2 v_{p+q}\varepsilon_q) D_{p+q} v_{p+q} \\
  &-  \frac 43\gamma^2(R+\gamma) \varepsilon_p \alpha( D_{p+q} v_{p+q}\varepsilon_q ) D_{p+q}^2 v_{p+q} \\
  & -   \varepsilon_p \alpha( D_{p+q} v_{p+q}\varepsilon_q )M_{p+q} v_{p+q}.
\end{aligned}\end{equation}

Finally, the last equation is equivalent to
\begin{equation}
\begin{aligned}
   \varepsilon_p M_{p+q} \alpha( v_{p+q}\varepsilon_q) D_{p+q} v_{p+q}&= \gamma(R+\frac\gamma3) \varepsilon_p \alpha( D_{p+q} v_{p+q}\varepsilon_q) D_{p+q} v_{p+q} \\
  & -  \frac 43\gamma^2(R+\gamma) \varepsilon_p \alpha( D_{p+q}^2 v_{p+q}\varepsilon_q) D_{p+q} v_{p+q}.
\end{aligned}\end{equation}

\begin{equation}\begin{aligned}\label{15.10.15:eq2}
   \varepsilon_p \Big( M_{p+q} \alpha( v_{p+q}\varepsilon_q) - \gamma(R+\frac\gamma3) \alpha(D_{p+q} v_{p+q}\varepsilon_q)  \\ +\frac 43\gamma^2(R+\gamma) \alpha(D_{p+q}^2 v_{p+q}\varepsilon_q) \Big) D_{p+q} v_{p+q} = 0.
\end{aligned}\end{equation}

Recall equation \eqref{box}, which states
\begin{equation}
  \Box u  =  \gamma(R+\frac\gamma3) u' - \frac 43\gamma^2(R+\gamma) u''.
\end{equation}
Apply this equation to all elements of basis $v_{p+q}$ and get
\begin{equation}\begin{aligned}
  M_{p+q}v_{p+q} - \gamma(R+\frac\gamma3) v_{p+q}' + \frac 43\gamma^2(R+\gamma) v_{p+q}'' &= 0,\\
  M_{p+q}v_{p+q} - \gamma(R+\frac\gamma3) D_{p+q} v_{p+q} + \frac 43\gamma^2(R+\gamma) D_{p+q}^2 v_{p+q} &=0.
\end{aligned}\end{equation}
Multiplying by $M_{p+q}^l$ from the left and $\varepsilon_q M_{p+q}^{n-1-l}$ from the right and summing over $l$ from $0$ to $n-1$ we get
\begin{equation}\begin{aligned}
  &M_{p+q}\alpha(v_{p+q}\varepsilon_q) - \gamma(R+\frac\gamma3) \alpha(D_{p+q} v_{p+q}\varepsilon_q) \\
  &+ \frac 43\gamma^2(R+\gamma) \alpha(D_{p+q}^2 v_{p+q}\varepsilon_q) =0.
\end{aligned}\end{equation}
At the end, multiplying by $\varepsilon_p$ from the left and $D_{p+q}v_{p+q}$ from the right completes the proof of the first part of the theorem.

As we have seen previously, $T$ and $Z$ are polynomials in $R$ of degree $p+q$. Let their coefficients be $T_j$ and $Z_j$ ($0 \leq j\leq p+q$) respectively. What we have proved so far implies
  \begin{equation}\begin{aligned}\label{tz}
    T_{p+q} + 4Z_{p+q} & =0, \\
    T_j + 4Z_j &= 4\gamma (j+1)Z_{j+1}, (j\leq 0 < p+q).
  \end{aligned}\end{equation}

Moreover, the above equations imply  that the systems $T_{p+q} = T_{p+q-1} = \ldots = T_0 = 0$ and $Z_{p+q} = Z_{p+q-1} = \ldots = Z_0=0$ are equivalent, i.e. equations \eqref{trace:2} and  $\eqref{eom:00:2}$ are equivalent.
\end{proof}

It remains to solve equation \eqref{trace:2}. That is a very difficult task and it can be done only for particular values of parameters $p$ and $q$. The following theorem gives three cases. All other possibilities for $q \leq p \leq 4$ are presented in appendix.

\begin{theo} \label{theo:part}
The equation \eqref{trace:2} is  satisfied in the following cases ($\kappa = \frac 1{16\pi G}$):
  \begin{itemize}
%
    \item [$p=2$, $q=1$:] $\FF(\gamma ) = \frac{9 \kappa  (\gamma +9 \Lambda )}{112 \gamma ^3}$, $\FF(2 \gamma )= \frac{3 \kappa  (\gamma +9 \Lambda )}{56 \gamma^3}$, $f_0= -\frac{\kappa  (4 \gamma +15 \Lambda )}{7 \gamma ^3}$, $\FF'(\gamma )= -\frac{3 \kappa  (\gamma +9 \Lambda )}{8 \gamma ^4}$,
    \item [$p=2$, $q=2$:] $\FF(\gamma ) = \frac{369 \kappa (\gamma +8 \Lambda )}{9344 \gamma ^4}$, $\FF(2 \gamma )= \frac{27 \kappa  (\gamma +8 \Lambda )}{4672\gamma ^4} $, $f_0= \frac{\kappa  (145 \gamma +576 \Lambda )}{876 \gamma ^4}$, $\FF'(\gamma )= -\frac{639 \kappa  (\gamma +8 \Lambda )}{2336\gamma ^5}$, $\FF'(2 \gamma )= -\frac{27 \kappa  (\gamma +8 \Lambda )}{9344 \gamma ^5}$,
    \item [$p=3$, $q=1$:] $\FF(\gamma )= \frac{\kappa  (107 \gamma +408 \Lambda )}{6432 \gamma ^4}$, $\FF(2 \gamma )= -\frac{\kappa  (173 \gamma +840 \Lambda )}{7504 \gamma ^4}$ ,$\FF(3 \gamma )= 0$, $f_0= -\frac{\kappa  (95 \gamma +768 \Lambda )}{268 \gamma ^4}$, $\FF'(\gamma )= -\frac{9 \kappa (\gamma +8 \Lambda )}{88 \gamma ^5}$.
  \end{itemize}
\end{theo}

\begin{proof}
  One can see that, for each of the values $p$ and $q$ listed in the theorem, all of the coefficients $T_j$  are linear combinations of the following $p+q+1$ "variables" $f_0 = \FF(0)$, $\FF(\gamma)$, $\ldots$, $\FF(p \gamma)$, $\FF'(\gamma)$, $\ldots$, $\FF'(q \gamma)$. Hence trace equation \eqref{trace} and $00$ equations are split into linear systems of $p+q+1$ equations with $p+q+1$ variables. Equation \eqref{tz} implies that these two systems are equivalent. Solving one of them (for example trace) in each of the three cases  gives the statement of the theorem.

In particular, for $p=2$ and $q=1$ the coefficients $T_j$ are given by:

\begin{equation}
\begin{aligned}
  T_0 &= \frac{4}{9} \left(-5 \FF'(\gamma ) \gamma ^4-32 \FF(\gamma) \gamma^3-19 \FF(2\gamma ) \gamma^3-3 f_0 \gamma ^3+9 \kappa \Lambda \right), \\
  T_1 &= \frac{1}{3} \left(10 \FF(\gamma ) \gamma^2 -55 \FF(2 \gamma ) \gamma ^2-9 f_0 \gamma ^2-3 \kappa \right),\\
  T_2 &= 2 \gamma  \left(10 \FF(\gamma ) - \FF(2 \gamma )+2 \gamma  \FF'(\gamma )\right),\\
  T_3 &= 3 \FF(2 \gamma )-2 \FF(\gamma ),\\
\end{aligned}
\end{equation}
The system that remains to be solved is
\begin{equation}
\begin{aligned}
-5 \FF'(\gamma ) \gamma ^4-32 \FF(\gamma) \gamma^3-19 \FF(2\gamma ) \gamma^3-3 f_0 \gamma ^3 &=-9 \kappa \Lambda, \\
10 \FF(\gamma ) \gamma^2 -55 \FF(2 \gamma ) \gamma ^2-9 f_0 \gamma ^2 &= 3 \kappa ,\\
10 \FF(\gamma ) - \FF(2 \gamma )+2 \gamma  \FF'(\gamma )&=0,\\
3 \FF(2 \gamma )-2 \FF(\gamma )&=0.
\end{aligned}
\end{equation}

The solution of this system is given in the theorem. The other cases are proved in the similar way.

\end{proof}

\section{Limit $M\to +\infty$}
Let the characteristic scale $M$ grow to infinity, the function $\FF\left(\frac\Box{M^2}\right)\to f_0$ and hemce the EOM \eqref{trace}, \eqref{eom:00} yield only conditions on $f_0$. In case $(p,q) \neq (1,1)$ there is no solution. The case $(p,q)=(1,1)$ provides the following equations 
\begin{equation}
  18\gamma f_0 = 3\kappa, \qquad 9\gamma^2 f_0 = -2 \kappa \Lambda.
\end{equation}
 
The solution of this system is $f_0 = \frac \kappa{6\gamma}$ if the cosmological constant $\Lambda$ is such that $\Lambda = - \frac \gamma{12}$.  Taking the further assumption $f_0 =0$ restores the General Relativity and we see that scale factor \eqref{scale_factor} is not a solution.
 
\section{Conclusion}

In this paper we have presented  cosmological bounce solution of the form $a(t) = a_0 \exp(-\frac\gamma{12}t^2)$. This solution is obtained in the modified gravity  model with nonlocal term $R^p \mathcal{F}(\Box) R^q$. In order to have a solution analytic function $\FF$ and its derivative $\FF'$ have to satisfy conditions of the form
\begin{equation}
  f_0= x_0,\quad \FF(k\gamma) = x_k\, (1 \leq k \leq p),\quad \FF'(l\gamma) = y_l\, (1 \leq l \leq q),
\end{equation}
for some constants $x_k$ and $y_l$. It is worth noting that  in all cases except $p=q=1$ the set of constants $x_k$ and $y_l$ is unique(for fixed values of $\gamma$ and $\Lambda$). Moreover, the case $p=q=1$ requires that constant $\gamma$ has special value ($\gamma = -12\Lambda$) which means that cosmological constant is required in order to have nontrivial solution. In the other cases there is no such restriction.

In the present paper we considered the model for particular values of the parameters $p$ and $q$. There is a possibility to extend some of these results to a general case. Matrices $M_p$ and $F_p$ can be defined for negative integer values of $p$, but it is much harder to do the computations since they are infinitely dimensional, and it is not clear if it will yield any solutions. For example, the model $p=-1$, $q=1$ was consider in \cite{BD1,DDGR3}.

\section{Appendix}

In a similar way as theorem \ref{theo:part} the following six cases are proved:

\begin{theo}
The equation \eqref{trace:2} is  satisfied in the following cases ($\kappa = \frac 1{16\pi G}$):
  \begin{itemize}
    \item [$p=3$, $q=2$:] $\FF(\gamma )= \frac{3 \kappa  (10702 \gamma +40497 \Lambda )}{245680 \gamma ^5}$, $\FF(2 \gamma )= -\frac{27 \kappa  (6 \gamma +25
   \Lambda )}{24568 \gamma ^5}$, $\FF(3 \gamma )= -\frac{27 \kappa  (6 \gamma +25 \Lambda )}{49136 \gamma ^5}$, $f_0= -\frac{3 \kappa  (7099
   \gamma +23949 \Lambda )}{15355 \gamma ^5}$, $\FF'(\gamma )= -\frac{3 \kappa  (11614 \gamma +68865 \Lambda )}{270248 \gamma ^6}$, $\FF'(2 \gamma
   )= \frac{513 \kappa  (6 \gamma +25 \Lambda )}{171976 \gamma ^6}$,
    \item [$p=3$, $q=3$:] $\FF(\gamma )= \frac{\kappa  (338597 \gamma +1847844 \Lambda )}{9513152 \gamma ^6}$,  $\FF(2 \gamma )= -\frac{21 \kappa  (379 \gamma +2076
   \Lambda )}{432416 \gamma ^6}$, $\FF(3 \gamma )= -\frac{9 \kappa  (379 \gamma +2076 \Lambda )}{864832 \gamma ^6}$, $f_0= \frac{9 \kappa  (77093
   \gamma +441108 \Lambda )}{1405352 \gamma ^6}$, $\FF'(\gamma )= -\frac{3 \kappa  (1462285 \gamma +8126148 \Lambda )}{13080584 \gamma ^7}$, $\FF'(2 \gamma )= \frac{255 \kappa  (379 \gamma +2076 \Lambda )}{1324274 \gamma ^7}$, $\FF'(3 \gamma )= \frac{3 \kappa  (379 \gamma +2076 \Lambda)}{432416 \gamma ^7}$,
    \item [$p=4$, $q=1$:] $\FF(\gamma )= \frac{3 \kappa  (1570 \gamma +11679 \Lambda )}{1800928 \gamma ^5}$, $\FF(2 \gamma )= -\frac{9 \kappa (50102 \gamma +262581 \Lambda )}{80141296 \gamma ^5}$, $\FF(3 \gamma )= -\frac{3 \kappa  (105430 \gamma +726207 \Lambda )}{49525520 \gamma ^5}$, $\FF(4 \gamma )= -\frac{3 \kappa  (1570 \gamma +11679 \Lambda )}{2251160 \gamma ^5}$, $f_0= -\frac{3 \kappa  (1111 \gamma +5361 \Lambda)}{56279 \gamma^5}$, $\FF'(\gamma )= -\frac{27 \kappa  (2 \gamma +15 \Lambda )}{2320 \gamma ^6}$,
    \item [$p=4$, $q=2$:] $\FF(\gamma )= \frac{27 \kappa  (116489 \gamma +976692 \Lambda )}{2128403200 \gamma ^6}$, $\FF(2 \gamma )= -\frac{27 \kappa  (4591 \gamma +19308 \Lambda )}{170272256 \gamma ^6}$, $\FF(3 \gamma )= \frac{9 \kappa  (2632969 \gamma +9126132 \Lambda )}{93649740800 \gamma ^6}$, $\FF(4\gamma )= 0$, $f_0= \frac{27 \kappa  (9773 \gamma +38204 \Lambda )}{6651260 \gamma ^6}$, $\FF'(\gamma )= -\frac{9 \kappa  (36889711 \gamma +230208108 \Lambda )}{22044176000 \gamma ^7}$, $\FF'(2 \gamma )= \frac{9 \kappa  (1257961 \gamma -26340492 \Lambda )}{108244505600 \gamma^7} $,
    \item [$p=4$, $q=3$:] $\FF(\gamma )= \frac{\kappa  (21007019473 \gamma +144494046423 \Lambda )}{2412863869280 \gamma ^7}$, $\FF(2 \gamma )= -\frac{3 \kappa (1426277827 \gamma +10110884265 \Lambda )}{1206431934640 \gamma ^7}$, $\FF(3 \gamma )= \frac{9 \kappa  (32585957 \gamma +237505338 \Lambda)}{4825727738560 \gamma ^7}$, $\FF(4 \gamma )= \frac{\kappa  (32585957 \gamma +237505338 \Lambda )}{1206431934640 \gamma ^7}$, \\
        $f_0= \frac{3\kappa  (3321165266 \gamma +25006112775 \Lambda )}{75401995915 \gamma ^7}$, $\FF'(\gamma )= -\frac{3 \kappa(43005362079625 \gamma +307070903674071 \Lambda )}{1924258935750800 \gamma ^8}$, $\FF'(2 \gamma )= \frac{\kappa  (15505640343740 \gamma +110842995690981 \Lambda)}{1503214190561440 \gamma ^8}$, $\FF'(3 \gamma )= -\frac{109 \kappa  (32585957 \gamma +237505338 \Lambda )}{26541502562080 \gamma^8}$,
    \item [$p=4$, $q=4$:] $\FF(\gamma )= \frac{3 \kappa  (37038228809 \gamma +146181469392 \Lambda )}{127137036761600 \gamma ^8}$, $\FF(2 \gamma )= -\frac{9 \kappa (238071667 \gamma +847503216 \Lambda )}{7803583635712 \gamma ^8}$, $\FF(3 \gamma )= \frac{537 \kappa  (765701 \gamma +2682288 \Lambda)}{9644878650880 \gamma ^8}$, $\FF(4 \gamma )= \frac{3 \kappa  (765701 \gamma +2682288 \Lambda )}{219201787520 \gamma ^8}$, $f_0= \frac{33 \kappa  (131820287 \gamma +420903432 \Lambda )}{343872804172 \gamma ^8}$, $\FF'(\gamma )= -\frac{81 \kappa  (261799491587 \gamma + 967343633136 \Lambda )}{4608717582608000 \gamma ^9}$, $\FF'(2 \gamma )= \frac{3 \kappa  (1392867522289 \gamma +4562593829712 \Lambda)}{6945189435783680 \gamma ^9}$, $\FF'(3 \gamma )= -\frac{4569 \kappa  (765701 \gamma +2682288 \Lambda )}{26523416289920 \gamma^9}$, $\FF'(4\gamma )= -\frac{9 \kappa  (765701 \gamma +2682288 \Lambda )}{876807150080 \gamma ^9}$.
  \end{itemize}
\end{theo}

\section*{Acknowledgement}
Work on this paper was supported by Ministry of Education, Science and Technological Development of the Republic of Serbia, grant No 174012.
I would like to thank Branko Dragovich, Alexey Koshelev, Zoran Rakic and Jelena Stankovic for useful discussions.
\section*{References}

\bibliography{etsquared}{}

\begin{thebibliography}{10}
\expandafter\ifx\csname url\endcsname\relax
  \def\url#1{\texttt{#1}}\fi
\expandafter\ifx\csname urlprefix\endcsname\relax\def\urlprefix{URL }\fi
\expandafter\ifx\csname href\endcsname\relax
  \def\href#1#2{#2} \def\path#1{#1}\fi

\bibitem{clifton}
T.~Clifton, P.~Ferreira, A.~Padilla, C.~Skordis, {Modified gravity and
  cosmology}, Phys.Rep. 513 (2012) 1--189.

\bibitem{odintsov}
S.~Nojiri, S.~D. Odintsov, {Unified cosmic history in modified gravity: from
  F(R) theory to Lorentz non-invariant models}, Phys.Rept. 505 (2011) 59--144.
\newblock \href {http://arxiv.org/abs/1011.0544} {\path{arXiv:1011.0544}},
  \href {http://dx.doi.org/10.1016/j.physrep.2011.04.001}
  {\path{doi:10.1016/j.physrep.2011.04.001}}.

\bibitem{Woodard:2014iga}
R.~Woodard, {Nonlocal Models of Cosmic Acceleration}, Found.Phys. 44 (2014)
  213--233.
\newblock \href {http://arxiv.org/abs/1401.0254} {\path{arXiv:1401.0254}},
  \href {http://dx.doi.org/10.1007/s10701-014-9780-6}
  {\path{doi:10.1007/s10701-014-9780-6}}.

\bibitem{biswas0}
T.~Biswas, A.~Mazumdar, W.~Siegel, {Bouncing universes in string-inspired
  gravity}, JCAP 0603 (2006) 009.
\newblock \href {http://arxiv.org/abs/hep-th/0508194}
  {\path{arXiv:hep-th/0508194}}, \href
  {http://dx.doi.org/10.1088/1475-7516/2006/03/009}
  {\path{doi:10.1088/1475-7516/2006/03/009}}.

\bibitem{biswas}
T.~Biswas, T.~Koivisto, A.~Mazumdar, {Towards a resolution of the cosmological
  singularity in non-local higher derivative theories of gravity}, JCAP 1011
  (2010) 008.
\newblock \href {http://arxiv.org/abs/1005.0590} {\path{arXiv:1005.0590}},
  \href {http://dx.doi.org/10.1088/1475-7516/2010/11/008}
  {\path{doi:10.1088/1475-7516/2010/11/008}}.

\bibitem{biswas:2011ar}
T.~Biswas, E.~Gerwick, T.~Koivisto, A.~Mazumdar, {Towards singularity and ghost
  free theories of gravity}, Phys.Rev.Lett. 108 (2012) 031101.
\newblock \href {http://arxiv.org/abs/1110.5249} {\path{arXiv:1110.5249}},
  \href {http://dx.doi.org/10.1103/PhysRevLett.108.031101}
  {\path{doi:10.1103/PhysRevLett.108.031101}}.

\bibitem{koshelev2}
A.~S. Koshelev, S.~Y. Vernov, {On bouncing solutions in non-local gravity},
  Phys.Part.Nucl. 43 (2012) 666--668.
\newblock \href {http://arxiv.org/abs/1202.1289} {\path{arXiv:1202.1289}},
  \href {http://dx.doi.org/10.1134/S106377961205019X}
  {\path{doi:10.1134/S106377961205019X}}.

\bibitem{DDGR1}
I.~Dimitrijevic, B.~Dragovich, J.~Grujic, Z.~Rakic, {On Modified Gravity},
  Springer Proc.Math.Stat. 36 (2013) 251--259.
\newblock \href {http://arxiv.org/abs/1202.2352} {\path{arXiv:1202.2352}},
  \href {http://dx.doi.org/10.1007/978-4-431-54270-4_17}
  {\path{doi:10.1007/978-4-431-54270-4_17}}.

\bibitem{DDGR2}
I.~Dimitrijevic, B.~Dragovich, J.~Grujic, Z.~Rakic, {New Cosmological Solutions
  in Nonlocal Modified Gravity}, Rom. J. Phys. 58~(5-6) (2013) 550--559.
\newblock \href {http://arxiv.org/abs/1302.2794} {\path{arXiv:1302.2794}}.

\bibitem{biswas:2013cha}
T.~Biswas, A.~Conroy, A.~S. Koshelev, A.~Mazumdar, {Generalized ghost-free
  quadratic curvature gravity}, Class.Quant.Grav. 31 (2014) 015022.
\newblock \href {http://arxiv.org/abs/1308.2319} {\path{arXiv:1308.2319}},
  \href {http://dx.doi.org/10.1088/0264-9381/31/1/015022,
  10.1088/0264-9381/31/15/159501} {\path{doi:10.1088/0264-9381/31/1/015022,
  10.1088/0264-9381/31/15/159501}}.

\bibitem{DDGR6}
I.~Dimitrijevic, B.~Dragovich, J.~Grujic, Z.~Rakic, Some cosmological solutions
  of a nonlocal modified gravity, Filomat 29~(3) (2015) 619--628.

\bibitem{DDGKR}
I.~Dimitrijevic, B.~Dragovich, J.~Grujic, A.~S. Koshelev, Z.~Rakic, {Cosmology
  of modified gravity with a non-local f(R)}\href
  {http://arxiv.org/abs/1509.04254} {\path{arXiv:1509.04254}}.

\bibitem{ruzmaikin}
{T.V. Ruzmaikina, A.A. Ruzmaikin}, {QUADRATIC CORRECTIONS TO THE LAGRANGIAN
  DENSITY OF THE GRAVITATIONAL FIELD AND THE SINGULARITY }, JETP 30 (1970) 372.

\bibitem{starobinsky1}
A.A.Starobinsky, {A New Type of Isotropic Cosmological Models Without
  Singularity}, Phys. Lett. B 91 (1980) 99--102.

\bibitem{starobinsky2}
A.A.Starobinsky, {}, Lect. Notes in Phys. 246 (1986) 1072.

\bibitem{Koshelev:2014voa}
A.~S. Koshelev, S.~{\relax Yu}. Vernov, {Cosmological Solutions in Nonlocal
  Models}, Phys. Part. Nucl. Lett. 11~(7) (2014) 960--963.
\newblock \href {http://arxiv.org/abs/1406.5887} {\path{arXiv:1406.5887}},
  \href {http://dx.doi.org/10.1134/S1547477114070255}
  {\path{doi:10.1134/S1547477114070255}}.

\bibitem{BD1}
B.~Dragovich, On nonlocal modified gravity and cosmology, in: Lie Theory and
  Its Applications in Physics, Vol. 111 of Springer Proceedings in Mathematics
  and Statistics, 2014, pp. 251--262.
\newblock \href {http://dx.doi.org/10.1007/978-4-431-55285-7_17}
  {\path{doi:10.1007/978-4-431-55285-7_17}}.

\bibitem{DDGR3}
I.~Dimitrijevic, B.~Dragovich, J.~Grujic, Z.~Rakic, A new model of nonlocal
  modified gravity, PUBLICATIONS DE L INSTITUT MATHEMATIQUE-BEOGRAD 94~(108)
  (2013) 187--196.

\end{thebibliography}

\end{document}